\documentclass{article}
\title{Lower bounds for constant query affine-invariant LCCs and LTCs}
\author{
Arnab Bhattacharyya\thanks{Research partially supported by a DST Ramanujan Fellowship.}\\
Indian Institute of Science\\
{\small \texttt{arnabb@csa.iisc.ernet.in}}
\and 
Sivakanth Gopi\thanks{Research partially supported by NSF grants CCF-1523816 and CCF-1217416}\\
Princeton University\\
{\small \texttt{sgopi@cs.princeton.edu}}
}

\date{}
\usepackage[pagebackref,colorlinks,linkcolor=blue,filecolor = blue,citecolor = blue, urlcolor = blue, hyperfootnotes=false]{hyperref}
\usepackage{ amsfonts, mathtools,bbm,amsthm,url,amssymb}
\usepackage{fullpage}
\usepackage{todonotes}

\newcounter{mynotes}
\newcommand{\polylog}{\mathrm{polylog}}
\newcommand{\inpro}[2]{\left\langle #1,#2 \right\rangle}
\newcommand{\norm}[1]{\|#1\|}
\newcommand{\Unorm}[2]{\|#1\|_{U^{#2}}}

\newcommand{\rk}{\mathrm{rank}}

\newcommand{\Supp}{\mathrm{Supp}}
\newcommand{\argmin}{\mathrm{argmin}}
\def\simplex{{\blacktriangle}}
\def\F{{\mathbb{F}}}

\def\Z{{\mathbb{Z}}}

\def\R{{\mathbb{R}}}
\def\K{{\mathbb{K}}}
\def\C{{\mathbb{C}}}

\def\E{{\mathbb E}}
\def\T{{\mathbb T}}

\def\cB{{\mathcal B}}

\def\cH{{\mathcal H}}
\def\cL{{\mathcal L}}

\def\cX{{\mathcal X}}
\def\cC{{\mathcal C}}
\def\cF{{\mathcal F}}
\def\cM{{\mathcal M}}
\def\cN{{\mathcal N}}
\def\cD{{\mathcal D}}

\def\tf{{\tilde{f}}}

\def\hf{{\hat{f}}}
\def\hg{{\hat{g}}}
\def\hD{{\widehat{\cD}}}
\def\hH{{\widehat{H}}}
\def\hh{{\hat{h}}}
\def\bf{\bar{f}}
\def\bg{\bar{g}}

\newcommand{\Tr}{{\rm Tr}}

\newcommand{\ignore}[1]{}
\newtheorem{lemma}{Lemma}
\newtheorem{theorem}{Theorem}

\newtheorem{definition}{Definition}
\newtheorem{remark}{Remark}
\begin{document}
\setcounter{page}{0}
\maketitle

\begin{abstract}
Affine-invariant codes are codes whose coordinates form a vector space over a finite field and which are invariant under affine transformations of the coordinate space. They form a natural, well-studied class of codes; they include popular codes such as Reed-Muller and Reed-Solomon. A particularly appealing feature of affine-invariant codes is that they seem well-suited to admit local correctors and testers. 

In this work, we give lower bounds on the length of locally correctable and locally testable affine-invariant codes with constant query complexity. We show that if a code $\cC \subset \Sigma^{\K^n}$ is an $r$-query locally correctable code (LCC), where $\K$ is a finite field and $\Sigma$ is a finite alphabet, then the number of codewords in $\cC$ is at most $\exp(O_{\K, r, |\Sigma|}(n^{r-1}))$. Also, we show that if $\cC \subset \Sigma^{\K^n}$ is an $r$-query locally testable code (LTC), then the number of codewords in $\cC$ is at most $\exp(O_{\K, r, |\Sigma|}(n^{r-2}))$. The dependence on $n$ in these bounds is tight for constant-query LCCs/LTCs, since Guo, Kopparty and Sudan (ITCS `13) construct affine-invariant codes via lifting that have the same asymptotic tradeoffs.   Note that our result holds for non-linear codes, whereas previously, Ben-Sasson and Sudan (RANDOM `11) assumed linearity to derive similar results.

Our analysis uses higher-order Fourier analysis. In particular, we show that the codewords corresponding to an affine-invariant LCC/LTC must be far from each other with respect to Gowers norm of an appropriate order. This then allows us to bound the number of codewords, using known decomposition theorems which approximate any bounded function in terms of a finite number of low-degree non-classical polynomials, upto a small error in the Gowers norm.
\end{abstract}
\newpage

\section{Introduction}

Error-correcting codes which admit local algorithms are of significant interest in theoretical computer science. A code is  called a \textsf{locally correctable code (LCC)} if there is a randomized algorithm that, given an index $i$ and a received word $w$ close to a codeword $c$ in Hamming distance, outputs $c_i$ by querying only a few positions of $w$. A code is called a \textsf{locally testable code (LTC)} if there is a randomized algorithm that, given a received word $w$, determines whether $w$ is in the code or whether $w$ is far in Hamming distance from every codeword, based on queries to a small number of locations of $w$.  The number of positions of the received word queried  is called the \textsf{query complexity} of the LCC or LTC.

The notions of local correctability and local testability have a long history in computer science by now. Also called ``self-correction'', the idea of local correction originated in works by Lipton \cite{Lipton90} and by Blum and Kannan \cite{BK89}  on program checkers. LCCs are closely related to \textsf{locally decodable codes (LDCs)}, where the goal is to recover a symbol of the underlying message when given a corrupted codeword, using a small number of queries \cite{KT00}. LDCs and LCCs have found applications in  {private information retrieval schemes} \cite{CGKM98, BIW07} and derandomization \cite{STV01}.  See \cite{Yek11} for a detailed survey on LDCs and LCCs. Research on LTCs implicitly started with Blum, Luby, and Rubinfeld's seminal discovery \cite{BLR93} that the Hadamard code is an LTC with query complexity $3$; they were first formally defined by Goldreich and Sudan in \cite{GS06}. LTCs have been used (implicitly and explicitly)  in many contexts, most notably in the construction of PCP's \cite{AS98, ALMSS98, Din07}. 

 In spite of the wide interest in them, some basic questions about LCCs and LTCs remain unanswered. We restrict ourselves throughout to the setting where the query complexity is a constant (independent of the length of the code) and consider the tradeoff between query complexity and code length. The current best constant-query LCCs have exponential length, while the current best constant-query LTCs have near-linear length but they are quite complicated \cite{BS08, Din07, Mei09, Vid15}. Getting subexponential length LCCs or linear length LTCs with constant query complexity are major open problems in the area.
 
Intuitively, for LCCs and LTCs with constant query complexity, there must be a lot of redundancy in the code, since every symbol of the codeword must satisfy local constraints with most other symbols in the codeword. A systematic way to generate redundancy is to make sure that the code has a large group of {\em invariances}\footnote{A quite different way to generate redundancy is through {\em tensoring}; see \cite{BS04}. Invariances and tensoring are essentially the only two ``generic'' reasons known to cause local correctability/testability.}. Formally, given a code $\cC \subset \Sigma^N$ of length $N$ over alphabet $\Sigma$, a codeword $c \in \cC$ can  be naturally viewed as a function $c : [N] \to \Sigma$. Then, we say that $\cC$ is \textsf{invariant} under a set $G \subset \{[N] \to [N]\}$\footnote{$\{A\to B\}$ and $B^A$ denote the set of all functions from $A$ to $B$. } if for every $\pi \in G$ and codeword $c \in \cC$, $c \circ \pi$ also describes a codeword $c' \in \cC$. Now, the key observation is that if for every codeword $c \in \cC$, if there is a constraint among $c(i_1), \dots, c(i_k)$ for some $i_1, \dots, i_k \in [N]$, then for every $c \in \cC$, there must also be a constraint among $c(\pi(i_1)), \dots, c(\pi(i_k))$ for any $\pi$ in the invariance set $G$, since $c \circ \pi$ is itself another codeword. Hence if $G$ is large, the presence of one local constraint immediately implies presence of many and suggests the possibility of local algorithms for the code. This connection between invariance and correctability/testability was first explicitly examined by Kaufman and Sudan \cite{KS08}. One is then motivated to understand more clearly the possibilities and limitations of local correctors/testers for codes possessing natural symmetries.
 
 We focus on \textsf{affine-invariant codes}, for which the domain $[N]$ is an $n$-dimensional vector space $\K^n$ over a finite field $\K$ and the code $\cC \subset \{\K^n \to \Sigma\}$ is invariant under affine transformations $A: \K^n \to \K^n$. Affine invariance is a very natural symmetry for ``algebraic codes'' and has long been studied in coding theory \cite{KLP67}. The study of affine-invariant LCCs and LTCs was initiated in \cite{KS08} and has been investigated in several follow-up works \cite{BS11, Guo13, BRS12, GSVW15}. The hope is that because affine-invariant codes have a large group of invariance and, at the same time, are conducive to non-trivial algebraic constructions, they may contain a code that improves current constructions of LCCs or LTCs.

The current best parameters for constant-query affine-invariant LCCs and LTCs are achieved by the lifted codes of Guo, Kopparty and Sudan \cite{GKS13}. They construct an affine-invariant code $\cF \subset \{\F_{2^\ell}^n \to \F_2\}$ with $\exp(\Theta(n^{r-2}))$ codewords that is an $(r-1)$-query LCC and an $r$-query LTC, where $r = 2^\ell$.  The $\Theta(\cdot)$ notation hides factors that depend on $r$ but not $n$. 
For LCCs, the same asymptotic tradeoff between query complexity and code length is achieved by the Reed-Muller code. For every $r \geq 2$, the Reed-Muller code of order $r-1$ (i.e., polynomials over $\F_q$ on $n$ variables of total degree $\leq r-1$ with $q>r$) is an affine-invariant $r$-query LCC with $\exp(\Theta(n^{r-1}))$ codewords. In fact, even if we drop the affine-invariance requirement, Reed-Muller codes and the construction of \cite{GKS13} achieve the best known codeword length for constant query LCCs\footnote{In contrast, there exist non-affine-invariant LTCs of constant query complexity and inverse polylogarithmic rate. This corresponds to an LTC with $\exp(N/\polylog(N))$ codewords, where $N$ is the code length, while the affine-invariant LTC of \cite{GKS13} and Reed-Muller codes have $\exp(\polylog(N))$ codewords for constant query complexity.}.

 In this work, we show that the parameters for the lifted codes of \cite{GKS13} are, in fact, tight for affine-invariant LCCs/LTCs in $\{\K^n \to \Sigma\}$ for any fixed finite field $\K$ and any fixed finite alphabet $\Sigma$.  
 
 \begin{theorem}[Main Result, informal]\label{main-informal}
~\\[-1.5em]
\begin{enumerate}
\item[(i)]
Let $\cC\subset \{\K^n \to \Sigma\}$ be an $r$-query affine-invariant LCC. Then $|\cC|\le\exp\left(O_{\K,r,|\Sigma|}(n^{r-1})\right)$.
\item[(ii)]
Let $\cC\subset \{\K^n \to \Sigma\}$ be an $r$-query affine-invariant LTC. Then $|\cC|\le\exp\left(O_{\K,r,|\Sigma|}(n^{r-2})\right)$.
\end{enumerate}
\end{theorem}

\subsection{Related Work}
 Ben-Sasson and Sudan in \cite{BS11} obtained a similar result as Theorem~\ref{main-informal}, when the code is assumed to be linear,  i.e., when the codewords form a vector space. They showed that if  $\cC \subset \{\K^n \to \F\}$ is an $(r-1)$-query locally correctable or $r$-query locally testable {\em linear}, affine-invariant code, where $\K$ and $\F$ are finite fields of characteristic $p>0$ with $\K$ an extension of $\F$, then the dimension of $\cC$ as a vector space over $\F$ is at most $(n \log_p |\K|)^{r-2}$. When $\K$ is fixed (as in \cite{GKS13}'s construction of constant query LCCs/LTCs), the result of \cite{BS11} is a very special case of our Theorem \ref{main-informal}. On the other hand, \cite{BS11}'s result also applies when the size of $\K$ is growing (as long as $\K$ extends $\F$), whereas ours does not.
 
There are several works which study lower bounds for constant query LCCs~\cite{KT00,GKST06,DS07,KdW03,BDYW11,BDSS11,Woo12,DSW14}. For general (non-affine-invariant) LCCs, tight lower bounds are known only for 2-query LCCs.
Kerendis and deWolf~\cite{KdW03} prove that if $\cC \subset \{\{0,1\}^n\to\Sigma\}$ is a 2-query LCC\footnote{Their lower bound also holds for the weaker notion of locally decodable dode (LDC)}, then $|\cC|\le \exp(O(n|\Sigma|^5)).$ 
This is tight for constant $\Sigma$ and achieved by the Hadamard code. For $r$-query LCCs where $r>2$, the lower bounds known are much weaker. 
The best known bounds, due to~\cite{KdW03,Woo07}, show that if $\cC \subset \{\{0,1\}^n\to\{0,1\}\}$ is an $r$-query LCC, then $$|\cC|\le \exp\left(2^{n/(1+1/(\lceil r/2 \rceil+1))+o(n)}\right).$$ Higher-order Fourier analysis was applied to other problems in coding theory in~\cite{BL15,TW14}.
 
\subsection{Proof Overview}
Our arguments are based on standard techniques from higher-order Fourier analysis \cite{Tao12}, but they are new in this context. We show that if an affine-invariant code is an $r$-query LCC, then its codewords are far from each other in the $U^r$-norm, the {\em Gowers norm of order $r$}. Similarly, we show that the codewords of an affine-invariant $r$-query LTC are far from each other in the $U^{r-1}$-norm. Therefore, we can upper bound the number of LCC/LTC codewords in terms of the size of a net that is fine enough with respect to the Gowers norm of an appropriate order.  We bound the size of such a net by explicitly constructing one using a standard decomposition theorem (analogous to Szemer\'edi's regularity lemma): any bounded function $f: \K^n \to \C$ can be approximated, upto a small error in the Gowers norm, by a composition of a bounded number of low-degree non-classical polynomials \cite{TZ12}. 

The way we argue that two codewords $f$ and $g$ of an $r$-query LCC are far in the Gowers norm is that if $\|f-g\|_{U^r} < \epsilon$, then for small enough $\epsilon$ (with respect to $r$, $|\Sigma|$ and correctness probability), the local corrector when applied to $f$ can act as if it is applied to $g$. The argument is, briefly, as follows. On the one hand, the codewords $f$ and $g$ must be far in Hamming distance, because the definition of LCC implies that there is a unique codeword close to any string. So, with constant probability over choice of $y \in \K^n$, the local corrector's guess for $f(y)$ must differ from $g(y)$. On the other hand, we can lower bound by a constant the probability of the event that the corrector outputs $g(y)$ when it queries coordinates of $f$, because $f$ and $g$ are close in the $\|\cdot\|_{U^r}$ norm. This last calculation uses the affine invariance of the code and the {\em generalized von Neumann inequality}, which bounds by $\|f_0\|_{U^{k}}$ the expectation over $z_1, \dots, z_m \in \K^n$ of the product $\prod_{i=0}^{k} f_i(\cL_i(z_1, \dots, z_m))$, where the $\cL_i$'s are arbitrary linear forms so that no two are linearly dependent and $f_i:\K^n\to \C$ are arbitrary functions with $|f_i|\le 1$.

 The argument for $r$-query LTCs is similar. Suppose $f$ and $g$ are close in the  $\|\cdot\|_{U^{r-1}}$ norm. Consider the random function $H$ such that for every $x$ independently, $H(x)$ equals $f(x)$ with probability $1/2$ and $g(x)$ with probability $1/2$. $H$ itself is far from a codeword with high probability. But we show that since the local tester accepts $f$, it will also accept $H \circ \ell$ for a random invertible affine map $\ell : \K^n \to \K^n$ with good probability. This implies that with good probability, $H\circ\ell$ is close to a codeword and by affine-invariance, $H$ itself is close to a codeword which gives a contradiction. To draw this conclusion, we again use the generalized von Neumann inequality as well as a hybrid argument.
 
 \paragraph{Organization.}
 Section \ref{sec:prelim} contains preliminaries that lay the foundations of our analysis. Section \ref{sec:lcc} proves the first part of our main result about LCCs, while Section \ref{sec:ltc} proves the second part about LTCs.
\section{Preliminaries}\label{sec:prelim}

\subsection{Error-correcting codes}
Let $\cX$ be a finite set called the set of coordinates and $\Sigma$ be an other finite set called the alphabet. Let $\Sigma^\cX$ denote the set of all functions from $\cX\to\Sigma$. A subset $\cC\subset \Sigma^\cX$ is called a code and its elements are called \textit{codewords}.

\begin{definition}[Hamming distance]
Given $f,g\in \Sigma^\cX$, we define the normalized Hamming distance between $f$ and $g$ is defined as $\Delta(f,g):=\Pr_{x\in \cX}[f(x)\ne g(x)]$ where $x$ is uniformly chosen from $\cX$. For a code $\cC\subset\Sigma^\cX$, we define the minimum distance of $\cC$ as $\min_{f,g\in \cC, f\ne g} \Delta(f,g)$.
\end{definition}

Let $\simplex_\Sigma=\{q:\Sigma\to \R_{\ge 0}:\sum_{i \in \Sigma}q(i)=1\}$ denote the probability simplex on $\Sigma$. We embed $\Sigma$ into $\simplex_\Sigma$ by sending $i\in \Sigma$ to $e_i$ which is the $i^{th}$ coordinate vector in $\R^{\Sigma}$. This also lets us extend functions $f:\cX\to\Sigma$ to $\hf:\cX\to \simplex_\Sigma$ using the embedding. We call $\hf$ the simplex extension of $f$. Now given $f,g\in \Sigma^\cX$, we can write the Hamming distance between them as $$\Delta(f,g)=1-\Pr_{x\in \cX}[f(x)=g(x)]=1-\E_{x\in \cX}\langle\hf,\hg\rangle$$ where $\inpro{\cdot}{\cdot}$ is the standard inner product in $\R^\Sigma$.

\begin{definition}[Affine invariance]
Let $\cX$ be a finite dimensional vector space over some finite field $\K$, then $\cC\subset \Sigma^\cX$ is called affine invariant if for every $f\in \cC$ and every invertible affine map $\ell:\cX\to \cX$, $f\circ \ell\in \cC$.
\end{definition}
Locally correctable and testable codes are defined formally in Sections \ref{sec:lcc} and \ref{sec:ltc} respectively.

\subsection{Higher order Fourier analysis}
Fix a finite field $\F_p$ of prime order $p$, and let $\K = \F_q$ where $q = p^t$ for a positive integer $t$. $\K$ is then a vector space of dimension $t$ over $\F_p$. We denote by $\Tr: \K \to \F_p$ the {\em trace function}:
$$\Tr(x) = x + x^p + x^{p^2} + \cdots + x^{p^{t-1}}.$$ Also, we use $|\cdot|$ to denote the obvious map from $\F_p$ to $\{0, 1, \dots, p-1\}$.

Given functions $f,g:\K^n\to \C$, we define their inner product as $\inpro{f}{g}=\E_x[\overline{f(x)}g(x)]$ where $x$ is chosen uniformly from $\K^n$. We define $\norm{\cdot}_p$-norm on such functions as $\norm{f}_p=\E_x[|f(x)|^p]^{1/p}$. We say a function $f:\K^n\to \C$ is {\em bounded} if $|f|\le 1$. Let $\T$ denote the circle group $\R/\Z$ and $e:\T\to\C$ be the map given by $e(x)=\exp(2\pi i x)$.

\begin{definition}[Non-classical Polynomials]
A {\em non-classical polynomial of degree $< d$} is a function $f:\K^n\to \T$ if $$\forall h_1,h_2\cdots,h_{d}\in \K^n\ \ D_{h_1}D_{h_2}\cdots D_{h_{d}}f=0$$ where $D_h$ is the difference operator defined as $D_hf(x)=f(x+h)-f(x)$.
For such an $f$, the function $e(f)$ is called a {\em non-classical phase polynomial of degree $<d$}.
\end{definition}
Let $\alpha_1,\cdots,\alpha_t\in \K$ be a basis for $\K$ when viewed as a vector space over $\F_p$. It is known \cite{TZ12, BB15} that non-classical polynomials of degree $\le d$ are exactly those functions $P: \K^n \to \T$ which have the following form:
\begin{equation}\label{eqn-frm}
P(x_1, \dots, x_n) = \theta + \sum_{k \geq 0} \sum_{\substack{0 \leq d_{i,j} < p~ \forall i \in [n], j \in [t]; \\ 0 < \sum_{i=1}^n \sum_{j=1}^t d_{i,j} \leq d-k(p-1)}} \frac{c_{d_{1,1}, \dots, d_{n,t}, k} \prod_{i=1}^n \prod_{j=1}^t |\Tr(\alpha_j x_i)|^{d_{i,j}}}{p^{k+1}} \pmod 1
\end{equation}
for some $c_{d_{1,1}, \dots, d_{n,t}, k}\in\{0,1,\cdots,p-1\}$ and $\theta\in \T$. Next, we define the Gowers norm for arbitrary functions $f: \K^n \to \C$.


\begin{definition}[Gowers uniformity norm~\cite{Gowers01}]
For a function $f:\K^n\to \C$, the {\em Gowers norm of order $r$}, denoted by $\Unorm{\cdot}{r}$, is defined as 
$$\Unorm{f}{r}=\left(\E_{x,h_1,\cdots,h_r\in \K^n}[\Delta_{h_1}\Delta_{h_2}\cdots \Delta_{h_r}f(x)]\right)^{1/2^r}$$ where $\Delta_{h}$ is the multiplicative difference operator defined as $\Delta_hf(x)=f(x+h)\overline{f(x)}$.
\end{definition}
The Gowers norm is an actual norm when $r\ge 2$. It also satisfies a useful monotonicity property: for any function $f: \K^n \to \C$,
$$|\E[f(x)]|=\Unorm{f}{1}\le \Unorm{f}{2}\le \cdots \le \Unorm{f}{r}\le \cdots \le \norm{f}_{\infty}.$$ See~\cite{Tao12} for more on Gowers norm. Observe that if $f:\K^n\to\C$ is a non-classical phase polynomial of degree $<r$ then $\Unorm{f}{r}=1$. The inverse Gowers theorem is a partial converse to this. It shows that the Gowers norm of order $r$ of a function is in direct correspondence with its correlation with non-classical phase polynomials of degree $<r$. In particular:
\begin{lemma}[Inverse Gowers theorem~\cite{TZ12}]
\label{lem-inversegowers}
For any bounded $f:\K^n\to \C$, if $\Unorm{f}{r}>\delta$ then there exists a non-classical polynomial $P$ of degree $< r$ such that $$|\inpro{f}{e(P)}|\ge c(\delta,\K,r)$$ where $c(\delta,\K,r)$ is a constant depending only on $\delta,\K,r$.
\end{lemma}

A linear form on $m$ variables is a vector $\cL=(w_1,\cdots,w_m)\in \K^m$ that is interpreted as a function $\cL:(\K^n)^m\to \K^n$ via the map $(x_1,\cdots,x_m)\mapsto \sum_{i=1}^m w_ix_i$. A key reason that the Gowers norm is useful in applications is that if a function has small Gowers norm of the appropriate order, then it behaves pseudorandomly in a certain way with respect to linear forms.
\begin{lemma}[Generalized von Neumann inequality (Exercise 1.3.23 in~\cite{Tao12})]
\label{lem-gowerscs}
Let $f_0,f_1,f_2,\cdots,f_k:\K^n \to \C$ be bounded functions and let $\cL=\{\cL_0,\cL_1,\cdots,\cL_k\}$ be a system of $k+1$ linear forms in $m$ variables such that no form is a multiple of another. Then
$$|\E_{z_1,\cdots,z_m\in \K^n}[\prod_{i=0}^k f_i(\cL_i(z_1,\cdots,z_m))]|\le \min_{0\le i\le k} \Unorm{f_i}{k}$$
\end{lemma}
See Appendix~\ref{sec-proofofgowerscs} for proof.

\subsection{A net for Gowers norm}
The goal of this section is to establish the following claim.
\begin{theorem}[$\epsilon$-net for $U^r$ norm]
\label{lem-epsilonnet}
The metric induced by the $\Unorm{\cdot}{r}$ norm on the space of all bounded functions $\{f:\K^n\to \C\}$ has an $\epsilon$-net of size $\exp(O_{\epsilon,\K,r}(n^{r-1}))$.
\end{theorem}

For the proof, we need the following definitions.
\begin{definition}[Polynomial factors]
A polynomial factor $\cB$ is a sequence of non-classical polynomials $P_1, . . . , P_k: \K^n \to \T$. We also identify it with the function $\cB : \K^n\to \T^k$ mapping $x\mapsto (P_1(x), . . . , P_k(x))$.  The partition induced by $\cB$ is the partition of $\K^n$ given by $\{\cB^{-1}(y):y \in \T^k\}$. The complexity of $\cB$ is the number of defining polynomials, $|\cB| = k$. The degree of $\cB$ is the maximum degree among its defining polynomials $P_1,\cdots, P_k$.
A function $f:\K^n\to \C$ is called $\cB$-measurable if it is constant in each cell of the partition induced by $\cB$ or equivalently $f$ can be written as a $\tau(P_1,\cdots,P_k)$ for some function $\tau:\T^k\to \C$.
\end{definition}

\begin{definition}[Conditional expectations]
Given a polynomial factor $\cB$, the conditional expectation of $f:\K^n\to \C$ over $\cB$, denoted by $\E[f|\cB]$, is the $\cB$-measurable function defined by $$\E[f|\cB](x)=\E_{y\in \cB^{-1}\left(\cB(x)\right)}[f(y)].$$
\end{definition}

\begin{definition}[Factor refinement]
Given two polynomial factors $\cB,\cB'$, we say $\cB'$ is a refinement of $\cB$, denoted by $\cB'\preceq\cB$, if every cell in the partition induced by $\cB'$ is contained in some cell in the partition induced by $\cB$.
\end{definition}
The definition of refinement immediately implies:
\begin{lemma}[Pythagoras theorem]
\label{lem-pythagoras}
Let $\cB,\cB'$ be polynomial factors such that $\cB'\preceq\cB$, then for any function $f:\K^n\to \C$,
$$\norm{\E[f|\cB']}_2^2=\norm{\E[f|\cB]}_2^2+\norm{\E[f|\cB']-\E[f|\cB]}_2^2.$$
\end{lemma}
\ignore{
\begin{proof}
It is enough to show that $\inpro{\E[f|\cB]}{\E[f|\cB']-\E[f|\cB]}=0$.
\begin{align*}
\inpro{\E[f|\cB]}{\E[f|\cB']-\E[f|\cB]}&=\E\Big[\E[f|\cB]\E[f|\cB']-\E[f|\cB]^2\Big]\\
&=\E\bigg[\E\Big[\E[f|\cB]\E[f|\cB']\Big|\cB\Big]\bigg]-\E\Big[\E[f|\cB]^2\Big]\\
&=\E\bigg[\E[f|\cB]\E\Big[\E[f|\cB']\Big|\cB\Big]\bigg]-\E\Big[\E[f|\cB]^2\Big]=0
\end{align*}
where in the last step we used the fact that $\cB'\preceq\cB\Rightarrow \E\Big[\E[f|\cB']\Big|\cB\Big]=\E[f|\cB]$.
\end{proof}}
The next claim shows that any bounded function is ``close'' to being measurable by a polynomial factor of bounded complexity. Precisely:
\begin{lemma}[Decomposition Theorem]
\label{lem-decomposition}
Any bounded $f:\K^n\to \C$ can be approximated in $\Unorm{\cdot}{r}$ by a function of a small number of degree $<r$ non-classical polynomials i.e. for any $\epsilon>0$, there exists non-classical polynomials $P_1,P_2,\cdots,P_k$ of degree $< r$ with $P_i(\bar{0})=0\ \forall i$ and a bounded function $\tau:\T^{k}\to \C$ such that
$$\Unorm{f-\tau(P_1,P_2,\cdots,P_k)}{r}\le \epsilon$$ where $k=k(\epsilon,\K,r)$ is a constant depending only on $\epsilon,\K,r$. 
\end{lemma}
\begin{proof}
The proof is similar to the proof of the Quadratic Koopman-von Neumann decompostion which is Prop 3.7 in~\cite{Green06} but using the full Inverse Gowers Theorem (Lemma~\ref{lem-inversegowers}) and similar claims are implicit elsewhere, but for completeness, we give the proof.

The main idea is to approximate the function $f$ using its conditional expectation over a suitable polynomial factor $\cB$ of degree $<r$. We will start with the trivial factor $\cB_0=(1)$ and iteratively construct more refined partitions $\cB_i\preceq \cB_{i-1}$ until we find a factor $\cB_k$ which satisfies $\Unorm{f-\E[f|\cB_k]}{r}\le \epsilon$. To bound the number of iterations needed to achieve this, we will show that the energy $\norm{\E[f|\cB_i]}_2^2$ which is bounded above by 1, increases by a fixed constant in every step.\\
Suppose that after step $i-1$, we still have $\Unorm{f-\E[f|\cB_{i-1}]}{r}>\epsilon$. Let $g=f-\E[f|\cB_{i-1}]$, then by the inverse Gowers theorem (Lemma~\ref{lem-inversegowers}), we have some non-classical polynomial $P_i$ of degree $<r$ such that $|\inpro{g}{e(P_i)}|\ge \kappa=c(\epsilon,p,r)$. We can assume that $P_i(\bar{0})=0$. Refine the factor $\cB_{i-1}$ by adding the polynomial $P_i$ to obtain $\cB_i\preceq \cB_{i-1}$. Now consider the energy increment,
$$\norm{\E[f|\cB_i]}_2^2-\norm{\E[f|\cB_{i-1}]}_2^2=\norm{\E[f|\cB_i]-\E[f|\cB_{i-1}]}_2^2=\norm{\E[g|\cB_i]}_2^2$$
where we used the Pythagoras theorem(Lemma~\ref{lem-pythagoras}) and the fact that $\E\big[\E[f|\cB_{i-1}]\big|\cB_i\big]=\E[f|\cB_{i-1}]$ since $\cB_i\preceq\cB_{i-1}$. So
\begin{align*}
\kappa^2 &\le |\E[g\cdot e(P_i)]|^2 = \left|\E\big[\E[g\cdot e(P_i)|\cB_i]\big]\right|^2=\left|\E\big[e(P_i)\E[g|\cB_i]\big]\right|^2\\
&\le \norm{\E[g|\cB_i]}_1^2\le\norm{\E[g|\cB_i]}_2^2 =\norm{\E[f|\cB_i]}_2^2-\norm{\E[f|\cB_{i-1}]}_2^2.
\end{align*}
Thus the energy increases by $\kappa^2$ every step. But since the energy is bounded above by 1, the process should end in a finite number of steps $k\le \frac{1}{\kappa^2}$. So $\Unorm{f-\E[f|\cB_k]}{r}\le \epsilon$, but since $\E[f|\cB_k]$ is $\cB_k$-measurable, we can write $\E[f|\cB_k]=\tau(P_1,\cdots,P_k)$ for some function $\tau$ with $|\tau|=|\E[f|\cB_k]|\le |f|\le 1$.
\end{proof}

We are now ready to prove Theorem \ref{lem-epsilonnet}.
\begin{proof}[Proof of Theorem \ref{lem-epsilonnet}]
Recall that $\K$ is an extension field of dimension $t$ over a prime field $\F_p$. The $\epsilon$-net will be the set $\cN$ of all functions of the form $\tau(P_1,\cdots,P_k)$ where $P_1,\cdots,P_k$ are degree $< r$ non-classical polynomials with zero constant terms, $\tau:\T^k\to \C$ is a bounded function and $k=k(\epsilon,p,r)$ is the constant given by Lemma~\ref{lem-decomposition}. But we will not include all possible bounded $\tau:\T^k\to \C$. Firstly by Equation \ref{eqn-frm}, $P_1,\cdots,P_k$ take values only in $\frac{1}{p^r}\Z/\Z$. Next we will discretize the set $\{z\in \C:|z|\le 1\}$ into the $\epsilon$-lattice i.e. we will only consider maps $\tau: (\frac{1}{p^r}\Z/\Z)^k \to \{z\in \C: |z|\le 1\}\cap \epsilon(\Z+i\Z)$. The number of such maps is bounded by $(4/\epsilon^2)^{p^{rk}}$.

By Equation \ref{eqn-frm}, a non-classical polynomial of degree $< r$ in $n$ variables with zero constant term can be represented by $\le \binom{nt+r-1}{r-1}r$ coefficients in $\{0,1,\cdots,p-1\}$. So the number of such non-classical polynomials is bounded by $\exp\left(O_{r,\K}(n^{r-1})\right)$. Combining both the bounds, $$|\cN|\le \exp\left(O_{r,\K}(n^{r-1})\right)^k \cdot (4/\epsilon^2)^{p^{rk}}=\exp\left(O_{\epsilon,\K,r}(n^{r-1})\right).$$

We will now prove that $\cN$ is a $3\epsilon$-net. Given any $f:\K^n\to [-1,1]$, using Lemma~\ref{lem-decomposition}, there is a function $\tau(P_1,\cdots,P_k)$ such that $$\Unorm{f-\tau(P_1,P_2,\cdots,P_k)}{r}\le \epsilon.$$
If we consider the $\tilde\tau\in \cN$ by rounding values real and imaginary parts of $\tau$ to the nearest multiple of $\epsilon$, we get 
\begin{align*}
\Unorm{f-\tilde\tau(P_1,P_2,\cdots,P_k)}{r}
&\le \Unorm{f-\tau(P_1,P_2,\cdots,P_k)}{r} + \Unorm{\tau(P_1,P_2,\cdots,P_k)-\tilde\tau(P_1,P_2,\cdots,P_k)}{r}\\
&\le \epsilon + \norm{\tau(P_1,P_2,\cdots,P_k)-\tilde\tau(P_1,P_2,\cdots,P_k)}_{\infty}\le 3\epsilon.
\end{align*}
\end{proof}


\section{Locally Correctable Codes}\label{sec:lcc}

We begin by defining locally correctable codes formally. Note that the definition below differs from the conventional one in terms of a local correction algorithm and adversarial errors (see, for instance, \cite{Yek11}); however, our definition is certainly weaker. Therefore, this makes our lower bounds  stronger.

\begin{definition}[Locally Correctable Code\ (LCC)]
An $(r,\delta,\tau)$ LCC is a code $\cC\subset \Sigma^\cX$ with the following property:\\
For each $x\in \cX$ there is a distribution $\cM_x$ over $r$-tuples of distinct\footnote{WLOG we can assume the tuples have distinct coordinates by adding dummy coordinates and modifying the decoding functions $\cD_{x,y_1,\cdots,y_r}$} coordinates  such that whenever $\tf\in \Sigma^\cX$ is $\delta$-close to some codeword $f\in\cC$ in Hamming distance,
$$\Pr_{(y_1,\cdots,y_r)\sim \cM_x}[\cD_{x,y_1,\cdots,y_r} (\tf(y_1),\tf(y_2),\cdots,\tf(y_r))=f(x)]\ge 1-\tau$$ where $\cD_{x,y_1,\cdots,y_r}:\Sigma^r\to \Sigma$, called the decoding operator, depends only on $x,y_1,\cdots,y_r$.\\
If furthermore $\cX$ is a vector space and $\cC$ is affine invariant then we call it an affine invariant LCC.
\end{definition}

\begin{remark}
Let $|\Sigma|=m$, WLOG we can assume that $\Sigma=\{1,2,\cdots,m\}$. Then we can extend functions $f:\cX\to \Sigma$ to $\hf:\cX\to \simplex_m$. The decoding operators $\cD:\Sigma^r\to\Sigma$ can also be extended to $\hD:\simplex_m^r\to\simplex_m$ as follows: For $z_1,\cdots,z_r\in \simplex_m$ define $$\hD(z_1,\cdots,z_r)=\sum_{1\le \ell_1,\cdots,\ell_r\le m} e_{\cD(\ell_1,\cdots,\ell_r)} (z_1)_{\ell_1}\cdots (z_r)_{\ell_r}$$ where $e_j$ stands for the $j^{th}$ coordinate vector in $\R^m$ and $(z_j)_\ell$ is the $\ell^{th}$ coordinate of the vector $z_j$. Now we can rewrite the decoding condition as:
$$\E_{(y_1,\cdots,y_r)\sim \cM_x}[\inpro{\hf(x)}{\hD_{x,y_1,\cdots,y_r} (\hf(y_1),\hf(y_2),\cdots,\hf(y_r))}]\ge 1-\tau.$$
\end{remark}

First, we make the observation that any LCC must have good minimum distance.
\begin{lemma}
\label{lem-distance}
Let $\cC\subset \Sigma^\cX$ be an $(r,\delta,\tau)$ LCC with $\tau<1/2$, then the minimum distance of $\cC$ is at least $2\delta$.
\end{lemma}
\begin{proof}
Let $f,g\in \cC$ be two distinct codewords such that $\Delta(f,g)<2\delta$. Let $h$ be the midpoint of $f$ and $g$ i.e. $h$ is $\delta$-close to both $f$ and $g$. Let $x\in\cX$ be such that $f(x)\ne g(x)$. By the LCC property, $$\Pr_{(y_1,\cdots,y_r)\sim \cM_x}[f(x)=\cD_{x,y_1,\cdots,y_r}(h(y_1),\cdots,h(y_r))]\ge 1-\tau$$
$$\Pr_{(y_1,\cdots,y_r)\sim \cM_x}[g(x)=\cD_{x,y_1,\cdots,y_r}(h(y_1),\cdots,h(y_r))]\ge 1-\tau.$$ This is a contradiction when $\tau<\frac{1}{2}$. Therefore every two codewords must be at least $2\delta$ apart.
\end{proof}

Now, we are ready to prove our main result of this section.
\begin{theorem}[Lower bound for LCCs]
\label{thm-main-LCC}
Let $\cC\subset \Sigma^{\K^n}$ be an $(r,\delta,\tau)$ affine-invariant LCC where $\tau< \frac{2\delta}{3}$. Then $|\cC|\le\exp\left(O_{\delta,\K,r,|\Sigma|}(n^{r-1})\right)$.
\end{theorem}
\begin{proof}
Let $|\Sigma|=m$. Let $\cN$ be an $\epsilon/2$-net for the space of all bounded functions $\{f:\K^n\to\C\}$ with the metric induced by $\Unorm{\cdot}{r}$-norm where $\epsilon=\frac{2\delta}{3m^r}$. Given a bounded $f:\K^n\to\C$, define $$\phi(f):=\argmin_{h\in \cN}\Unorm{f-h}{r}$$ (break ties arbitrarily). Since $\cN$ is an $\epsilon/2$ net, we have $\Unorm{f-\phi(f)}{r}\le \epsilon/2$. Define $\Psi:\cC\to \cN^m$ as $$\Psi(f):=(\phi(\hf_1),\cdots,\phi(\hf_m))$$ where $\hf_i:\K^n\to \R_{\ge 0}$ is the $i^{th}$ coordinate function of the simplex extension $\hf:\K^n\to \simplex_m$ of $f$. We claim that $\Psi$ is one-one which implies that $|\cC|\le |\cN|^m$. Now using Theorem~\ref{lem-epsilonnet}, the required bound follows. 
Suppose that $\Psi$ is not one-one. Let $f,g\in\cC$ be two distinct codewords such that $\Psi(f)=\Psi(g)$. This implies that $$\forall\ i\in [m]\ \Unorm{\hf_i-\hg_i}{r}\le \Unorm{\hf_i-\phi(\hf_i)}{r}+\Unorm{\hg_i-\phi(\hg_i)}{r} \le \epsilon.$$ 
By affine invariance of $\cC$, $f\circ \ell \in \cC$ for all invertible affine maps $\ell:\K^n \to \K^n$. So by the local correction property, $$\Pr_{\ell,y_0,(y_1,\cdots,y_r)\sim \cM_{y_0}}[f\circ \ell (y_0)=\cD_{y_0,y_1,\cdots,y_r} (f\circ \ell(y_1),\cdots,f\circ \ell(y_r))]\ge 1-\tau$$
where $\ell$ ranges uniformly over all invertible affine maps from $\K^n\to\K^n$ and $y_0$ ranges uniformly over $\K^n$.
Now consider the following difference:
\begin{align*}
&\Pr_{\ell,y_0,(y_1,\cdots,y_r)\sim \cM_{y_0}}[f\circ \ell (y_0)=\cD_{y_0,y_1,\cdots,y_r} (f\circ \ell(y_1),\cdots,f\circ \ell(y_r))]\\
&\hskip 2cm -\Pr_{\ell,y_0,(y_1,\cdots,y_r)\sim \cM_{y_0}}[g\circ \ell (y_0)=\cD_{y_1,\cdots,y_r} (f\circ \ell(y_1),\cdots,f\circ \ell(y_r))]\\
&=\E_{\ell}\E_{y_0}\E_{(y_1,\cdots,y_r)\sim \cM_{y_0}}\left[\inpro{\hf\circ \ell (y_0)}{\hD_{y_0,y_1,\cdots,y_r} (\hf\circ \ell(y_1),\cdots,\hf\circ \ell(y_r))}\right.\\
&\hskip 4cm \left.-\inpro{\hg\circ \ell (y_0)}{\hD_{y_1,\cdots,y_r} (\hf\circ \ell(y_1),\cdots,\hf\circ \ell(y_r))}\right]\\
&= \E_{y_0}\E_{(y_1,\cdots,y_r)\sim \cM_{y_0}}\left[\E_{\ell}\left[\inpro{\hf\circ \ell (y_0)-\hg\circ \ell (y_0)}{\hD_{y_0,y_1,\cdots,y_r} (\hf\circ \ell(y_1),\cdots,\hf\circ \ell(y_r))}\right]\right].
\end{align*}
Now we fix $y_0,y_1,\cdots,y_r$ and show that inner expectation is small for each tuple $(y_0,y_1,\cdots,y_r)$. Let us denote $\cD=\cD_{y_0,y_1,\cdots,y_r}$ for brevity. Let $t=\rk(y_0,y_1,\cdots,y_r)$, thus there exist independent vectors $v_1,\cdots,v_t\in \K^n$ such that for every $0\le i\le r$, $y_i=\sum_{j=1}^t \lambda_{ij}v_j$ for some fixed $\lambda_{ij}\in \K$. The action of a random invertible affine map $\ell$ can be approximated by sampling $z_0,z_1,\cdots,z_t\in \K^n$ uniformly and mapping $y_i\mapsto  z_0+\sum_{j=1}^t \lambda_{ij}z_j$ since with probability $1-o_n(1)$, $z_1,\cdots,z_t$ will be independent. Therefore,
 \begin{align*}
&\E_{\ell}\left[\inpro{\hf\circ \ell (y_0)-\hg\circ \ell (y_0)}{\hD_{y_0,y_1,\cdots,y_r} (\hf\circ \ell(y_1),\cdots,\hf\circ \ell(y_r))}\right]\\
&=_{o_n(1)} \E_{z_0,z_1,\cdots,z_t\in \K^n}\left[\inpro{(\hf-\hg)(z_0+\sum_{j=1}^t \lambda_{0j}z_j)}{\hD\left(\hf(z_0+\sum_{j=1}^t \lambda_{1j}z_j),\cdots,\hf(z_0+\sum_{j=1}^t \lambda_{rj}z_j)\right)}\right]\\
&=\E_{z_0,z_1,\cdots,z_t\in \K^n}\left[\inpro{(\hf-\hg)(z_0+\sum_{j=1}^t \lambda_{0j}z_j)}{\left(\sum_{1\le \ell_1,\cdots,\ell_r\le m}e_{\cD(\ell_1,\cdots,\ell_r)}\prod_{i=1}^r \hf_{\ell_i}(z_0+\sum_{j=1}^t \lambda_{ij}z_j)\right)}\right]\\
&=\E_{z_0,z_1,\cdots,z_t\in \K^n}\left[\left(\sum_{1\le \ell_1,\cdots,\ell_r\le m} (\hf-\hg)_{\cD(\ell_1,\cdots,\ell_r)}(z_0+\sum_{j=1}^t \lambda_{0j}z_j)\cdot \prod_{i=1}^r \hf_{\ell_i}(z_0+\sum_{j=1}^t \lambda_{ij}z_j)\right)\right]\\
&\le\left(\sum_{0\le \ell_1,\cdots,\ell_r\le m-1}\Unorm{(\hf-\hg)_{\cD(\ell_1,\cdots,\ell_r)}}{r}\right) 
\le m^r \epsilon
\end{align*}
 where the first inequality is obtained by applying generalized von Neumann inequality (Lemma~\ref{lem-gowerscs}) to each term. Therefore
\begin{align*}
&\Pr_{\ell,y_0,(y_1,\cdots,y_r)\sim \cM_{y_0}}\left[g\circ \ell (y_0)=\cD_{y_1,\cdots,y_r} (f\circ \ell(y_1),\cdots,f\circ \ell(y_r))\right]\\
&\ge \Pr_{\ell,y_0,(y_1,\cdots,y_r)\sim \cM_{y_0}}\left[f\circ \ell (y_0)=\cD_{y_1,\cdots,y_r} (f\circ \ell(y_1),\cdots,f\circ \ell(y_r))\right] - m^r\epsilon \ge 1-\tau-2\delta/3.
\end{align*}
On the other hand,
\begin{align*}
&\Pr_{\ell,y_0,(y_1,\cdots,y_r)\sim \cM_{y_0}}\left[g\circ \ell (y_0)=\cD_{y_1,\cdots,y_r} (f\circ \ell(y_1),\cdots,f\circ \ell(y_r))\right]\\
&\le \Pr_{\ell,y_0,(y_1,\cdots,y_r)\sim \cM_{y_0}}\left[g\circ \ell (y_0)=f\circ \ell(y_0)\right]+\Pr_{\ell,y_0,(y_1,\cdots,y_r)\sim \cM_{y_0}}\left[f\circ \ell (y_0)\ne\cD_{y_1,\cdots,y_r} (f\circ \ell(y_1),\cdots,f\circ \ell(y_r))\right]\\
&\le \Pr_x[f(x)=g(x)]+\tau  \le 1-2\delta+\tau \tag{By Lemma~\ref{lem-distance}}
\end{align*}
This is a contradiction when $\tau < \frac{2\delta}{3}$.

\end{proof}

 \section{Locally Testable Codes}\label{sec:ltc}
 We start by defining locally testable codes in a formulation convenient for our use. 
\begin{definition}[Locally Testable Code\ (LTC)]
An $(r,\delta,\tau)$ LTC is a code $\cC\subset \Sigma^\cX$ with minimum distance at least $\delta$ and the following property:\\
There is a distribution $\cM$ over $r$-tuples of distinct\footnote{WLOG we can assume the tuples have distinct coordinates by adding dummy coordinates and modifying the decoding functions $\cD_{y_1,\cdots,y_r}$} coordinates such that for each codeword $f\in C$, $$\Pr_{(y_1,\cdots,y_r)\sim \cM}[\cD_{y_1,\cdots,y_r} (f(y_1),f(y_2),\cdots,f(y_r))=1]\ge 3/4$$ and for every $g\in \Sigma^{\cX}$ which is $\tau$-far away from every codeword, $$\Pr_{(y_1,\cdots,y_r)\sim \cM}[\cD_{y_1,\cdots,y_r} (g(y_1),f(y_2),\cdots,f(y_r))=1]\le 1/4$$ where $\cD_{y_1,\cdots,y_r}:\Sigma^r\to \{0,1\}$, called the testing operator, depends only on $y_1,\cdots,y_r$. \\
If furthermore $\cX$ is a vector space and $\cC$ is affine-invariant then we call it an affine invariant LTC.
\end{definition}

\begin{remark}
Let $|\Sigma|=m$, WLOG we can assume that $\Sigma=\{1,2,\cdots,m\}$.  We can extend $f:\cX\to \Sigma$ to $\hf:\cX\to \simplex_m$. The testing operator $\cD:\Sigma^r\to\{0,1\}$ can also be extended to $\hD:\simplex_m^r\to [0,1]$ as follows: For $z_1,\cdots,z_r\in \simplex_m$ define 
\begin{equation} \label{eqn-multilinear-exp}
\hD(z_1,\cdots,z_r)=\sum_{1\le \ell_1,\cdots,\ell_r\le m} \cD(\ell_1,\cdots,\ell_r) (z_1)_{\ell_1}\cdots (z_r)_{\ell_r}.
\end{equation}
Now we can rewrite the probability in terms of expectation as:
$$\Pr_{(y_1,\cdots,y_r)\sim \cM}[\cD_{y_1,\cdots,y_r} (f(y_1),\cdots,f(y_r))=1]=\E_{(y_1,\cdots,y_r)\sim \cM}[\hD_{y_1,\cdots,y_r} (\hf\circ \ell(y_1),\cdots,\hf\circ \ell(y_r))]$$
\end{remark}
We are now ready to prove the main result of this section.
\begin{theorem}[Lower bound for LTC's]
\label{thm-main-LTC}
Let $\cC\subset \Sigma^{\K^n}$ be an $(r,\delta,\delta/3)$ affine invariant LTC, then $|\cC|\le\exp\left(O_{\delta,\K,r,|\Sigma|}(n^{r-2})\right)$.
\end{theorem}
\begin{proof}
Let $|\Sigma|=m$. The proof is very similar to that of Theorem~\ref{thm-main-LCC}. Let $\cN$ be an $\epsilon/2$-net for the space of all bounded functions $\{f:\K^n\to\C\}$ with the metric induced by $\Unorm{\cdot}{r-1}$-norm where $\epsilon=1/2rm^r$.  Define $\Psi:\cC\to \cN^m$ as in the proof of Theorem~\ref{thm-main-LCC}, it is enough to show that $\Psi$ is one-one.  
Suppose that $\Psi$ is not one-one. Then there exists $f,g\in\cC$ which are distinct such that $\Psi(f)=\Psi(g)$. This implies that $$\forall\ i\in [m]\ \Unorm{\hf_i-\hg_i}{r-1} \le \epsilon.$$ 

By affine invariance of $\cC$, $f\circ \ell \in \cC$ for all invertible affine maps $\ell:\K^n \to \K^n$. So  $$\E_{\ell}\E_{(y_1,\cdots,y_r)\sim \cM}[\cD_{y_1,\cdots,y_r} (f\circ \ell(y_1),f\circ \ell(y_2),\cdots,f\circ \ell(y_r))]\ge 3/4$$
where $\ell$ ranges over all invertible affine maps from $\K^n\to\K^n$. Let $H\in \Sigma^\cX$ be a random word where for each coordinate $x\in \cX$ independently,
\[
H(x)=\begin{cases}
f(x) & w.p.\ 1/2\\
g(x) & w.p.\ 1/2
\end{cases}
\]

Define $\hh:\cX\to \simplex_m$ as $\hh(x)=\E_H[\hH(x)]=\frac{\hf(x)+\hg(x)}{2}$ where $\hf,\hg$ are the simplex extensions of the original $f,g$. So $\forall\ i\in[m]\ \Unorm{\hf_i-\hh_i}{r-1}=\Unorm{\hf_i-\hg_i}{r-1}/2\le \epsilon/2$. We will now show that the test accepts $H\circ\ell$ with good probability when $\ell$ is a random invertible affine map from $\K^n\to\K^n$.
\begin{align*}
&\E_H\E_{\ell}\E_{(y_1,\cdots,y_r)\sim \cM}[\cD_{y_1,\cdots,y_r} (f\circ \ell(y_1),\cdots,f\circ \ell(y_r))-\cD_{y_1,\cdots,y_r} (H\circ \ell(y_1),\cdots,H\circ \ell(y_r))]\\
&=\E_H\E_{\ell}\E_{(y_1,\cdots,y_r)\sim \cM}[\hD_{y_1,\cdots,y_r} (\hf\circ \ell(y_1),\cdots,\hf\circ \ell(y_r))-\hD_{y_1,\cdots,y_r} (\hH\circ \ell(y_1),\cdots,\hH\circ \ell(y_r))]\\
&=\E_{\ell}\E_{(y_1,\cdots,y_r)\sim \cM}[\hD_{y_1,\cdots,y_r} (\hf\circ \ell(y_1),\cdots,\hf\circ \ell(y_r))-\hD_{y_1,\cdots,y_r} (\hh\circ \ell(y_1),\cdots,\hh\circ \ell(y_r))] \tag{by using the multilinear expansion of $\hD_{y_1,\cdots,y_r}$(Equation~\ref{eqn-multilinear-exp}) and taking expectation over $H$}\\
&= \E_{(y_1,\cdots,y_r)\sim \cM}\left[\E_{\ell}\left[\hD_{y_1,\cdots,y_r} (\hf\circ \ell(y_1),\cdots,\hf\circ \ell(y_r))-\hD_{y_1,\cdots,y_r} (\hh\circ \ell(y_1),\cdots,\hh\circ \ell(y_r))\right]\right] 
\end{align*}
Now we fix $y_1,\cdots,y_r$ and show that inner expectation is small for each tuple $(y_1,\cdots,y_r)$. Let us denote $\cD=\cD_{y_1,\cdots,y_r}$ for brevity. Let $t=\rk(y_1,\cdots,y_r)$, thus there exist independent vectors $v_1,\cdots,v_t\in \K^n$ such that for every $1\le i\le r$, $y_i=\sum_{j=1}^t \lambda_{ij}v_j$ for some fixed $\lambda_{ij}\in \K$. The action of a random invertible affine map $\ell$ can be approximated by sampling $z_0,z_1,\cdots,z_t\in \K^n$ uniformly and mapping $y_i\mapsto  z_0+\sum_{j=1}^t \lambda_{ij}z_j$ since with probability $1-o_n(1)$, $z_1,\cdots,z_t$ will be independent. Therefore,
 \begin{align*}
&\E_{\ell}\left[\hD_{y_1,\cdots,y_r} (\hf\circ \ell(y_1),\cdots,\hf\circ \ell(y_r))-\hD_{y_1,\cdots,y_r} (\hh\circ \ell(y_1),\cdots,\hh\circ \ell(y_r))\right]\\
&=_{o_n(1)}\E_{z_0,\cdots,z_t\in \K^n}\left[\hD(\hf(z_0+\sum_{j=1}^t \lambda_{1j}z_j),\cdots,\hf(z_0+\sum_{j=1}^t \lambda_{rj}z_j))-\cD(\hh(z_0+\sum_{j=1}^t \lambda_{1j}z_j),\cdots,\hh(z_0+\sum_{j=1}^t \lambda_{rj}z_j))\right]\\
&=\E_{z_0,z_1,\cdots,z_t\in \K^n}\left[\sum_{1\le \ell_1,\cdots,\ell_r\le m}\cD(\ell_1,\cdots,\ell_r)\left(\prod_{i=1}^r \hf_{\ell_i}(z_0+\sum_{j=1}^t \lambda_{ij}z_j)-\prod_{i=1}^r \hh_{\ell_i}(z_0+\sum_{j=1}^t \lambda_{ij}z_j)\right)\right]\\
&\le r\cdot m^r\cdot \frac{\epsilon}{2}=\frac{1}{4}
 \end{align*}
 where the last line is obtained by forming hybrids i.e. writing $$\hf_{\ell_1}\cdot \hf_{\ell_2} \cdots \hf_{\ell_r}- \hh_{\ell_1}\cdot \hh_{\ell_2} \cdots \hh_{\ell_r}=(\hf_{\ell_1}-\hh_{\ell_1})\cdot \hf_{\ell_2} \cdots \hf_{\ell_r}+\hh_{\ell_1}\cdot (\hf_{\ell_2}-\hh_{\ell_2}) \cdots \hf_{\ell_r}+\cdots +\hh_{\ell_1}\cdot \hh_{\ell_2} \cdots (\hf_{\ell_r}-\hh_{\ell_r})$$ and using Lemma~\ref{lem-gowerscs} for each term. Therefore
\begin{align*}
&\E_H\E_{\ell}\E_{(y_1,\cdots,y_r)\sim \cM}[\cD_{y_1,\cdots,y_r} (H\circ \ell(y_1),\cdots,H\circ \ell(y_r))]\\
&\ge \E_{\ell}\E_{(y_1,\cdots,y_r)\sim \cM}[\cD_{y_1,\cdots,y_r} (f\circ \ell(y_1),\cdots,f\circ \ell(y_r))] -\frac{1}{4} \ge \frac{3}{4}-\frac{1}{4}=\frac{1}{2}.
\end{align*}
 By Markov inequality, 
\begin{align*}
\frac{1}{4}&\le \Pr_H\left[\E_{\ell}\E_{(y_1,\cdots,y_r)\sim \cM}[\cD_{y_1,\cdots,y_r} (H\circ \ell(y_1),\cdots,H\circ \ell(y_r))]\ge \frac{1}{3}\right]\\
&\le \Pr_H\left[\exists \ell\ \ \E_{(y_1,\cdots,y_r)\sim \cM}[\cD_{y_1,\cdots,y_r} (H\circ \ell(y_1),\cdots,H\circ \ell(y_r))]\ge \frac{1}{3}\right]\\
&\le \Pr_H\left[\exists \ell\ \Delta(H\circ \ell, \cC)]\le \frac{\delta}{3}\right]\tag{by the soundness of the tester}\\
&= \Pr_H\left[ \Delta(H, \cC)]\le \frac{\delta}{3}\right]\tag{since $\ell$ is invertible and $\cC$ is affine invariant}\\ 
\end{align*} 
Let $\cH=\Supp(H)$ be the set of words between $f$ and $g$ i.e. the set of all words $e\in \Sigma^{\K^n}$ such that $e(x)=f(x)$ or $e(x)=g(x)$ for all $x\in\K^n$. We have $|\cH|=2^{\Delta(f,g)n}$. Since the distribution of $H$ is uniform in $\cH$, we proved that at least $\frac{1}{4}$ fraction of words in $\cH$ contain a codeword in their $\delta/3$ neighborhood, let $\cH'\subset \cH$ denote this subset. Therefore the $\delta/6$ neighborhoods around the points in $\cH'$ must be disjoint or else two distinct codewords will be $<\delta$ close to each other. The number of words in $\cH$ which lie in a Hamming ball of radius $\delta/6$ around a point of $\cH'$ is $$\sum_{i=0}^{\delta n/6}\binom{\Delta(f,g)n}{i}\ge 2^{H(\delta/6\Delta(f,g))\Delta(f,g)n-o(n)}\ge 2^{H(\delta/6)\Delta(f,g)n-o(n)}$$ where $H(\cdot)$ is the binary entropy function. By a packing argument, we can upper bound the size of $\cH'$ as $$|\cH'|\le \frac{2^{\Delta(f,g)n}}{2^{H(\delta/6)\Delta(f,g)n-o(n)}}=o(|\cH|).$$ This contradicts the fact that $|\cH'|\ge |\cH|/4$.
 

\end{proof}

\section{Concluding Remarks}
In this work, we proved tight lower bounds for constant query affine-invariant LCCs and LTCs when the number of queries $r$, underlying field $\K$ and the alphabet $\Sigma$ are constant. However the constants in the bounds we obtain are of Ackermann-type in $r,|\K|,|\Sigma|$ because of the use of higher-order Fourier analysis. Improving the dependence on these parameters is an open problem which might require new ideas. In a recent work, Bhowmick and Lovett~\cite{BL15_1} obtain a ``bias implies low rank" theorem for polynomials over growing fields. This might be a first step towards proving a variant of the inverse Gowers theorem (Lemma~\ref{lem-inversegowers}) for growing field size, which could then be used to make our lower bounds extend to the case of growing field size.


We also remark that our lower bounds work for any LCC or LTC where the queries are obtained as fixed linear combinations of uniformly chosen points from $\K^n$. Affine-invariant codes are a natural class of local codes where this is true. Relaxing these conditions to get lower bounds for a more general class of LCCs or LTCs is an open problem.
\section*{Acknowledgements}
We thank Madhu Sudan for helpful pointers to previous work. The second author would like to thank his advisor, Zeev Dvir, for his guidance and encouragement.
\bibliographystyle{alpha}
\bibliography{references}


\appendix
\section{Proof of generalized von Neumann inequality (Lemma~\ref{lem-gowerscs})}
\label{sec-proofofgowerscs}
Since the lemma is not stated in the form we want in~\cite{Tao12}, we will include a proof here for completeness. To prove Lemma~\ref{lem-gowerscs}, we need the following lemma first.
\begin{lemma}[Exercise 1.3.22 in~\cite{Tao12}]
\label{lem-gowerscs-basic}
Let $f:\K^n\to \C$ be a function, and for each $1\le i\le k$, let $g_i:(\K^n)^k\to\C$ be a bounded function which is independent of the $i^{th}$ coordinate of $(\K^n)^k$. Then, $$|\E_{x_1,\cdots,x_k\in \K^n}[f(x_1+x_2+\cdots+x_k)\prod_{i=1}^k g_i(x_1,\cdots,x_k)]|\le \Unorm{f}{k}$$
\end{lemma}
\begin{proof}
The proof is by induction on $k$ and using Cauchy-Schwarz inequality repeatedly. The case $k=1$ is true by definition of $\Unorm{\cdot}{1}$.
\begin{align*}
&\left|\E_{x_1,\cdots,x_k\in \K^n}\left[f(x_1+x_2+\cdots+x_k)\prod_{i=1}^k g_i(x_1,\cdots,x_k)\right]\right|\\
&=\left|\E_{x_2,\cdots,x_k}\left[g_1(x_1,\cdots,x_k)\E_{x_1}\left[f(x_1+x_2+\cdots+x_k)\prod_{i=2}^k g_i(x_1,\cdots,x_k)\right]\right]\right| \tag{since $g_1$ doesn't depend on $x_1$}\\
&\le \left|\E_{x_2,\cdots,x_k}\left[\E_{x_1'}\left[f(x_1'+x_2+\cdots+x_k)\prod_{i=2}^k g_i(x_1',x_2,\cdots,x_k)\right]\E_{x_1}\left[\bf(x_1+x_2+\cdots+x_k)\prod_{i=2}^k \bg_i(x_1,x_2,\cdots,x_k)\right]\right]\right|^{1/2} \tag{By Cauchy-Schwarz inequality and the fact that $|g_1|\le 1$}\\
&= \left|\E_{x_1,h_1}\left[\E_{x_2,\cdots,x_k}\left[\Delta_{h_1}f(x_1+x_2+\cdots+x_k)\prod_{i=2}^k g_i(x_1+h_1,x_2,\cdots,x_k)\bg_i(x_1,x_2,\cdots,x_k)\right]\right]\right|^{1/2}\tag{By substituting $x_1'=x_1+h_1$}\\
&\le \left|\E_{x_1,h_1}\left[\E_{h_2,\cdots,h_k,z}\left[\Delta_{h_k}\cdots\Delta_{h_1}f(x_1+z)\right]^{1/2^{k-1}}\right]\right|^{1/2} \tag{By induction hypothesis and the definition of Gowers norm}\\
&\le \left|\E_{x_1,h_1,h_2,\cdots,h_k,z}\left[\Delta_{h_k}\cdots\Delta_{h_1}f(x_1+z)\right]\right|^{1/2^k} \tag{By Jensen's inequality}\\
& =\left|\E_{h_1,h_2,\cdots,h_k,z}\left[\Delta_{h_k}\cdots\Delta_{h_1}f(z)\right]\right|^{1/2^k}=\Unorm{f}{k}
\end{align*}
\end{proof}

\begin{proof}[Proof of Lemma~\ref{lem-gowerscs}]
By symmetry, it is enough to show that $$|\E_{z_1,\cdots,z_m\in \K^n}[ f_0(\cL_0(z_1,\cdots,z_m))\prod_{i=1}^k f_i(\cL_i(z_1,\cdots,z_m))]|\le  \Unorm{f_0}{k}.$$ We will make a linear change of variables so that we can use Lemma~\ref{lem-gowerscs-basic} to get the required bound. For each $1\le i \le k$, since $\cL_0$ is not a multiple of $\cL_i$, there exists a vector $v_i\in \K^m$ such that $\cL_0(v_i)=1$ and $\cL_i(v_i)=0$.  Now we make the following change of variables: $(z_1,\cdots,z_m)\rightarrow (x_1,\cdots,x_m)+\sum_{i=1}^k y_iv_i^T$ where $x_1,\cdots,x_m$ and $y_1,\cdots,y_k$ are the new variables which range over $\K^n$. 
\begin{align*}
&|\E_{z_1,\cdots,z_m\in \K^n}[ f_0(\cL_0(z_1,\cdots,z_m))\prod_{i=1}^k f_i(\cL_i(z_1,\cdots,z_m))]|\\
&=\left|\E_{x_1,\cdots,x_m,y_1,\cdots,y_k\in \K^n}\left[f_0\left(\cL_0(x_1,\cdots,x_m)+\sum_{j\in [k]}y_j\right)\prod_{i\in [k]}f_i\left(\cL_i(x_1,\cdots,x_m)+\sum_{j\in [k]\setminus\{i\}}y_j\cL_i(v_j)\right)\right]\right| \tag{By change of variables and linearity of $\cL_i$}\\
&\le \E_{x_1,\cdots,x_m\in \K^n}\left[\left|\E_{y_1,\cdots,y_k\in \K^n}\left[f_0\left(\cL_0(x_1,\cdots,x_m)+\sum_{j\in [k]}y_j\right)\prod_{i\in [k]} f_i\left(\cL_i(x_1,\cdots,x_m)+\sum_{j\in [k]\setminus\{i\}}y_j\cL_i(v_j)\right)\right]\right|\right] \\
&\le  \Unorm{f_0}{k} \tag{By Lemma~\ref{lem-gowerscs-basic}}
\end{align*}
\end{proof}

\end{document}